\documentclass[oneside,12pt,]{article}
\usepackage{geometry}
\usepackage{ifthen}
\usepackage{etoolbox}
\usepackage{ifxetex,ifluatex}
\usepackage{graphicx}
\PassOptionsToPackage{dvipsnames,svgnames,table}{xcolor}
\usepackage{xcolor}
\usepackage{tcolorbox}
\tcbuselibrary{skins}
\tcbuselibrary{breakable}
\tcbuselibrary{raster}
\tcbset{ bwminimalstyle/.style={size=minimal, boxrule=-0.3pt, frame empty,
colback=white, colbacktitle=white, coltitle=black, opacityfill=0.0} }
\tcbset{ runintitlestyle/.style={fonttitle=\blocktitlefont\upshape\bfseries, attach title to upper} }
\tcbset{ exercisespacingstyle/.style={before skip={1.5ex plus 0.5ex}} }
\tcbset{ blockspacingstyle/.style={before skip={2.0ex plus 0.5ex}} }
\tcbuselibrary{xparse}
\usetikzlibrary{calc}
\tcbuselibrary{hooks}
\newlength{\normalparindent}
\newlength{\normalparskip}
\AtBeginDocument{\setlength{\normalparindent}{\parindent}}
\AtBeginDocument{\setlength{\normalparskip}{\parskip}}
\newcommand{\setparstyle}{\setlength{\parindent}{\normalparindent}\setlength{\parskip}{\normalparskip}}
\IfFileExists{latexrelease.sty}{}{\usepackage{fixltx2e}}
\ifthenelse{\boolean{xetex} \or \boolean{luatex}}{%
\ifxetex\usepackage{xltxtra}\fi
\ifluatex\usepackage{realscripts}\fi
}{
\usepackage[utf8]{inputenc}
}
\newcommand{\divisionfont}{\relax}
\newcommand{\blocktitlefont}{\relax}
\newcommand{\contentsfont}{\relax}

\newcommand{\xreffont}{\relax}

\ifthenelse{\boolean{xetex} \or \boolean{luatex}}{%
\usepackage{fontspec}
\usepackage{polyglossia}
\setmainlanguage[variant=usmax]{english}
}{%
\usepackage{lmodern}
\usepackage[T1]{fontenc}
}
\usepackage{microtype}
\usepackage{amsmath}
\usepackage{amscd}
\usepackage{amssymb}
\allowdisplaybreaks[4]
\usepackage[explicit, newparttoc, pagestyles]{titlesec}
\usepackage{titletoc}
\newcommand{\divisionnameptx}{\relax}%
\newcommand{\authorsptx}{\relax}%
\NewDocumentEnvironment{sectionptx}{mmmmmmm}
{%
\renewcommand{\divisionnameptx}{#1}%
\renewcommand{\titleptx}{#2}%
\renewcommand{\subtitleptx}{#3}%
\renewcommand{\shortitleptx}{#4}%
\renewcommand{\authorsptx}{#5}%
\renewcommand{\epigraphptx}{#6}%
\section[{#4}]{#2}%
\label{#7}%
}{}%
\NewDocumentEnvironment{subsectionptx}{mmmmmmm}
{%
\renewcommand{\divisionnameptx}{#1}%
\renewcommand{\titleptx}{#2}%
\renewcommand{\subtitleptx}{#3}%
\renewcommand{\shortitleptx}{#4}%
\renewcommand{\authorsptx}{#5}%
\renewcommand{\epigraphptx}{#6}%
\subsection[{#4}]{#2}%
\label{#7}%
}{}%
\NewDocumentEnvironment{references-section}{mmmmmmm}
{%
\renewcommand{\divisionnameptx}{#1}%
\renewcommand{\titleptx}{#2}%
\renewcommand{\subtitleptx}{#3}%
\renewcommand{\shortitleptx}{#4}%
\renewcommand{\authorsptx}{#5}%
\renewcommand{\epigraphptx}{#6}%
\section[{#4}]{#2}%
\label{#7}%
}{}%
\NewDocumentEnvironment{references-section-numberless}{mmmmmmm}
{%
\renewcommand{\divisionnameptx}{#1}%
\renewcommand{\titleptx}{#2}%
\renewcommand{\subtitleptx}{#3}%
\renewcommand{\shortitleptx}{#4}%
\renewcommand{\authorsptx}{#5}%
\renewcommand{\epigraphptx}{#6}%
\section*{#2}%
\addcontentsline{toc}{section}{#4}
\label{#7}%
}{}%
\titleformat{\part}[display]
{\divisionfont\Huge\bfseries\centering}{\divisionnameptx\space\thepart}{30pt}{\Huge#1}
[{\Large\centering\authorsptx}]
\titleformat{\chapter}[display]
{\divisionfont\huge\bfseries}{\divisionnameptx\space\thechapter}{20pt}{\Huge#1}
[{\Large\authorsptx}]
\titleformat{name=\chapter,numberless}[display]
{\divisionfont\huge\bfseries}{}{0pt}{#1}
[{\Large\authorsptx}]
\titlespacing*{\chapter}{0pt}{50pt}{40pt}
\titleformat{\section}[hang]
{\divisionfont\Large\bfseries}{\thesection}{1ex}{#1}
[{\large\authorsptx}]
\titleformat{name=\section,numberless}[block]
{\divisionfont\Large\bfseries}{}{0pt}{#1}
[{\large\authorsptx}]
\titlespacing*{\section}{0pt}{3.5ex plus 1ex minus .2ex}{2.3ex plus .2ex}
\titleformat{\subsection}[hang]
{\divisionfont\large\bfseries}{\thesubsection}{1ex}{#1}
[{\normalsize\authorsptx}]
\titleformat{name=\subsection,numberless}[block]
{\divisionfont\large\bfseries}{}{0pt}{#1}
[{\normalsize\authorsptx}]
\titlespacing*{\subsection}{0pt}{3.25ex plus 1ex minus .2ex}{1.5ex plus .2ex}
\titleformat{\subsubsection}[hang]
{\divisionfont\normalsize\bfseries}{\thesubsubsection}{1em}{#1}
[{\small\authorsptx}]
\titleformat{name=\subsubsection,numberless}[block]
{\divisionfont\normalsize\bfseries}{}{0pt}{#1}
[{\normalsize\authorsptx}]
\titlespacing*{\subsubsection}{0pt}{3.25ex plus 1ex minus .2ex}{1.5ex plus .2ex}
\titleformat{\paragraph}[hang]
{\divisionfont\normalsize\bfseries}{\theparagraph}{1em}{#1}
[{\small\authorsptx}]
\titleformat{name=\paragraph,numberless}[block]
{\divisionfont\normalsize\bfseries}{}{0pt}{#1}
[{\normalsize\authorsptx}]
\titlespacing*{\paragraph}{0pt}{3.25ex plus 1ex minus .2ex}{1.5em}
\titlecontents{part}%
[0pt]{\contentsmargin{0em}\addvspace{1pc}\contentsfont\bfseries}%
{\Large\thecontentslabel\enspace}{\Large}%
{}%
[\addvspace{.5pc}]%
\titlecontents{chapter}%
[0pt]{\contentsmargin{0em}\addvspace{1pc}\contentsfont\bfseries}%
{\large\thecontentslabel\enspace}{\large}%
{\hfill\bfseries\thecontentspage}%
[\addvspace{.5pc}]%
\dottedcontents{section}[3.8em]{\contentsfont}{2.3em}{1pc}%
\dottedcontents{subsection}[6.1em]{\contentsfont}{3.2em}{1pc}%
\dottedcontents{subsubsection}[9.3em]{\contentsfont}{4.3em}{1pc}%

\newcommand{\terminology}[1]{\textbf{#1}}
\newcommand{\pubtitle}[1]{\textsl{#1}}
\numberwithin{equation}{section}
\usepackage{enumitem}
\newlist{referencelist}{description}{4}
\setlist[referencelist]{leftmargin=!,labelwidth=!,labelsep=0ex,itemsep=1.0ex,topsep=1.0ex,partopsep=0pt,parsep=0pt}
\usepackage[hyperfootnotes=false]{hyperref}
\hypersetup{hidelinks}
\hypersetup{hypertexnames=false}
\makeatletter
\patchcmd\Hy@EveryPageBoxHook{\Hy@EveryPageAnchor}{\Hy@hypertexnamestrue\Hy@EveryPageAnchor}{}{\fail}
\makeatother
\hypersetup{pdftitle={Correctness of Extended RSA Public Key Cryptosystem}}

\newtcolorbox[auto counter, number within=section]{block}{}
\newtcolorbox[auto counter, number within=section]{project-distinct}{}
\makeatletter
\newtcolorbox[auto counter, number within=tcb@cnt@block, number freestyle={\noexpand\thetcb@cnt@block(\noexpand\alph{\tcbcounter})}]{subdisplay}{}
\makeatother
\tcbset{ theoremstyle/.style={bwminimalstyle, runintitlestyle, blockspacingstyle, after title={\space}, before upper app={\setparstyle}, } }
\newtcolorbox[use counter from=block]{theorem}[4]{title={{#1~\thetcbcounter\notblank{#2#3}{\space}{}\notblank{#2}{\space#2}{}\notblank{#3}{\space(#3)}{}}}, phantomlabel={#4}, breakable, after={\par}, fontupper=\itshape, theoremstyle, }
\tcbset{ corollarystyle/.style={bwminimalstyle, runintitlestyle, blockspacingstyle, after title={\space}, before upper app={\setparstyle}, } }
\newtcolorbox[use counter from=block]{corollary}[4]{title={{#1~\thetcbcounter\notblank{#2#3}{\space}{}\notblank{#2}{\space#2}{}\notblank{#3}{\space(#3)}{}}}, phantomlabel={#4}, breakable, after={\par}, fontupper=\itshape, corollarystyle, }
\tcbset{ proofstyle/.style={bwminimalstyle, fonttitle=\blocktitlefont\itshape, attach title to upper, after title={\space}, after upper={\space\space\hspace*{\stretch{1}}\(\blacksquare\)},
} }
\newtcolorbox{proof}[3]{title={\notblank{#2}{#2}{#1.}}, phantom={\label{#3}\hypertarget{#3}{}}, breakable, after={\par}, proofstyle, before upper app={\setparstyle} }
\tcbset{ definitionstyle/.style={bwminimalstyle, runintitlestyle, blockspacingstyle, after title={\space}, after upper={\space\space\hspace*{\stretch{1}}\(\lozenge\)}, before upper app={\setparstyle}, } }
\newtcolorbox[use counter from=block]{definition}[3]{title={{#1~\thetcbcounter\notblank{#2}{\space\space#2}{}}}, phantomlabel={#3}, breakable, after={\par}, definitionstyle, }
\tcbset{ remarkstyle/.style={bwminimalstyle, runintitlestyle, blockspacingstyle, after title={\space}, before upper app={\setparstyle}, } }
\newtcolorbox[use counter from=block]{remark}[3]{title={{#1~\thetcbcounter\notblank{#2}{\space\space#2}{}}}, phantomlabel={#3}, breakable, after={\par}, remarkstyle, }
\tcbset{ examplestyle/.style={bwminimalstyle, runintitlestyle, blockspacingstyle, after title={\space}, after upper={\space\space\hspace*{\stretch{1}}\(\square\)}, before upper app={\setparstyle}, } }
\newtcolorbox[use counter from=block]{example}[3]{title={{#1~\thetcbcounter\notblank{#2}{\space\space#2}{}}}, phantomlabel={#3}, breakable, after={\par}, examplestyle, }
\NewDocumentEnvironment{introduction}{m}
{\notblank{#1}{\noindent\textbf{#1}\space}{}}{\par\medskip}
\usepackage{tikz}
\usepackage{extpfeil}

\title{Correctness of Extended RSA Public Key Cryptosystem}
\author{Dar-jen Chang\\
University of Louisville\\
\href{mailto:djchan01@louisville.edu}{\nolinkurl{djchan01@louisville.edu}}
\and
Suranjan Gautam\\
University of Louisville\\
\href{mailto:suranjan.gautam@louisville.edu}{\nolinkurl{suranjan.gautam@louisville.edu}}
}
\date{}
\begin{document}
\raggedbottom
\label{shorttitlelowercase}\hypertarget{shorttitlelowercase}{}
\maketitle
\thispagestyle{empty}
\renewcommand*{\abstractname}{Abstract}
\begin{abstract}
This paper proposes an alternative approach to formally establishing the correctness of the RSA public key cryptosystem. The methodology presented herein deviates slightly from conventional proofs found in existing literature. Specifically, this study explores the conditions under which the choice of the positive integer N, a fundamental component of RSA, can be extended beyond the standard selection criteria. We derive explicit conditions that determine when certain values of N are valid for the encryption scheme and explain why others may fail to satisfy the correctness requirements. The scope of this paper is limited to the mathematical proof of correctness for RSA-like schemes, deliberately omitting issues related to the cryptographic security of RSA.%
\end{abstract}
\typeout{************************************************}
\typeout{Section 1 Introduction and Definitions}
\typeout{************************************************}
\begin{sectionptx}{Section}{Introduction and Definitions}{}{Introduction and Definitions}{}{}{section-1}
\begin{introduction}{}%
This section introduces the mathematical framework for our extended RSA encryption scheme and establishes the correctness conditions that will be analyzed throughout this paper.%
\end{introduction}%
\typeout{************************************************}
\typeout{Subsection 1.1 Problem Statement}
\typeout{************************************************}
\begin{subsectionptx}{Subsection}{Problem Statement}{}{Problem Statement}{}{}{subsection-1}
\begin{introduction}{}%
To achieve our goals, we begin by defining an extended framework for the RSA public key encryption system.%
\end{introduction}%
\begin{definition}{Definition}{Extended RSA Public Key Encryption Scheme.}{def-generalized-rsa}%
Given a positive integer \(N\) and \(\phi(N)\) denoting the Euler's totient function, we choose an integer \(e\) such that \(1 \leq e < \phi(N)\) and \(\gcd(e, \phi(N)) = 1\). We then compute the multiplicative inverse of \(e\) modulo \(\phi(N)\) and denote it as \(d\), where \(ed \equiv 1 \pmod{\phi(N)}\).%
\par
For any integer \(m\), \(1 \leq m < N\), the encryption and decryption functions are defined as follows:%
\par
\terminology{Encryption Function (\(Enc\)):} \(c = Enc(m) = m^e \pmod N\)%
\par
\terminology{Decryption Function (\(Dec\)):} \(m = Dec(c) = c^d \pmod N\)%
\end{definition}
\begin{remark}{Remark}{}{subsection-1-4}%
The pair \((e, N)\) forms the public key, while \((d, N)\) forms the private key.%
\end{remark}
\begin{definition}{Definition}{Correctness of Generalized RSA Public Key Encryption Scheme.}{def-generalized-rsa-2}%
The correctness condition for this encryption scheme is described as follows:%
\par
For any integer \(m\), \(1 \leq m < N\), we have:%
\begin{equation}
Dec(Enc(m)) \equiv m \pmod{N}\label{eq-correctness}
\end{equation}
or equivalently,%
\begin{equation}
(m^e)^d \equiv m \pmod{N}\label{def-generalized-rsa-2-2-5}
\end{equation}
\end{definition}
The classical (or textbook) RSA algorithm chooses \(N\) as the product of two large, distinct prime numbers, which guarantees that \hyperref[eq-correctness]{({\xreffont\ref{eq-correctness}})} holds. However, this paper seeks to investigate the generalized criteria under which \hyperref[eq-correctness]{({\xreffont\ref{eq-correctness}})} remains valid for various choices of \(N\). To achieve this, we provide a modified proof of correctness that builds upon existing proofs in the literature and derive necessary criteria for the selection of \(N\).%
\end{subsectionptx}
\end{sectionptx}
\typeout{************************************************}
\typeout{Section 2 Preliminary Number Theory Background}
\typeout{************************************************}
\begin{sectionptx}{Section}{Preliminary Number Theory Background}{}{Preliminary Number Theory Background}{}{}{section-2}
\begin{introduction}{}%
In this section, we will cover some basic concepts of number theory such as divisibility, fundamental theorem of arithmetic, greatest common divisor, and modular arithmetic, which are important for understanding RSA.%
\end{introduction}%
\typeout{************************************************}
\typeout{Subsection 2.1 Integer Division and Fundamental theorem of Arithmetic}
\typeout{************************************************}
\begin{subsectionptx}{Subsection}{Integer Division and Fundamental theorem of Arithmetic}{}{Integer Division and Fundamental theorem of Arithmetic}{}{}{subsection-2-1}
Integer division and remainder are fundamental operations in modular arithmetic. Given two integers \(a\) and \(b\), where \(b \neq 0\), the integer division of \(a\) by \(b\) is denoted as \(q = a \div b\), and the remainder is denoted as \(r = a \bmod b\). These values satisfy the equation:%
\par
\(a = b \cdot q + r\), where \(0 \leq r < b\).%
\par
If the remainder \(r\) is zero, then \(a\) is divisible by \(b\), and we say that \(b\) divides \(a\) and is denoted by \(b|a\).%
\begin{theorem}{Theorem}{Fundamental theorem of Arithmetic.}{}{Factorization-Theorem}%
Every integer \(N\) greater than 1 can be uniquely expressed as a product of prime factors:%
\begin{equation}
N = p_1^{n_1} p_2^{n_2} \cdots p_k^{n_k} = \prod_{i=1}^{k} p_i^{n_i}\label{Factorization-Theorem-2-1-2}
\end{equation}
where \(p_1,p_2,\ldots,p_k\) are distinct prime numbers and \(n_1,n_2,\ldots,n_k\) are their corresponding positive integer exponents.%
\end{theorem}
If%
\begin{equation*}
N = p_1^{n_1} p_2^{n_2} \cdots p_k^{n_k}
\end{equation*}
and \(a|N\), then \(a\) can be expressed as a product of the prime factors:%
\begin{equation*}
a = p_1^{m_1} p_2^{m_2} \cdots p_k^{m_k}
\end{equation*}
where \(0 \leq m_i \leq n_i\) (usually if \(m_i = 0\), \(p_i\) is not included in the product). This means that \(a\) can be formed by taking some of the prime factors of \(N\) and raising them to powers that do not exceed the corresponding powers in the prime factorization of \(N\).%
\end{subsectionptx}
\typeout{************************************************}
\typeout{Subsection 2.2 The Greatest Common Divisor and Co-Prime}
\typeout{************************************************}
\begin{subsectionptx}{Subsection}{The Greatest Common Divisor and Co-Prime}{}{The Greatest Common Divisor and Co-Prime}{}{}{subsection-2-2}
The greatest common divisor (gcd) of two integers \(a\) and \(b\), denoted as \(\gcd(a, b)\), is the largest positive integer that divides both \(a\) and \(b\) without leaving a remainder.%
\par
Two integers \(a\) and \(b\) are said to be co-prime (or relatively prime) if their gcd is 1, i.e., \(\gcd(a, b) = 1\). This implies that \(a\) and \(b\) have no common positive integer factors other than 1.%
\par
Many useful results in number theory and cryptography rely on the properties of co-prime integers. For example, in the RSA public key cryptosystem, the choice of the public exponent \(e\) must be co-prime to the Euler totient function \(\phi(N)\) to ensure that the encryption and decryption processes work correctly. To prove the correctness of the RSA algorithm, we need the following two theorems:%
\begin{theorem}{Theorem}{Product of Divisors.}{}{product-of-divisors}%
For any three integers \(a\), \(b\), and \(N\), if \(a|N\), \(b|N\) and \(\gcd(a, b) = 1\), then \(ab|N\).%
\end{theorem}
\begin{proof}{Proof}{}{product-of-divisors-3}
Assume%
\begin{equation*}
N = p_1^{n_1} p_2^{n_2} \cdots p_k^{n_k}
\end{equation*}
as given in the Fundamental theorem of Arithmetic. Since \(a|N\) and \(b|N\), we can express \(a\) and \(b\) in terms of the prime factors of \(N\):%
\begin{equation*}
a = p_1^{m_1} p_2^{m_2} \cdots p_k^{m_k}
\end{equation*}
and%
\begin{equation*}
b = p_1^{l_1} p_2^{l_2} \cdots p_k^{l_k}
\end{equation*}
where \(0 \leq m_i \leq n_i\) and \(0 \leq l_i \leq n_i\). Since \(\gcd(a, b) = 1\), it follows that for each prime factor \(p_i\), either \(m_i = 0\) or \(l_i = 0\). Therefore, the product \(ab\) can be expressed as:%
\begin{equation*}
ab = p_1^{m_1 + l_1} p_2^{m_2 + l_2} \cdots p_k^{m_k + l_k}
\end{equation*}
where \(m_i + l_i \leq n_i\) for each \(i\). This implies that \(ab\) is a product of the prime factors of \(N\) raised to powers that do not exceed the corresponding powers in the prime factorization of \(N\). Hence, \(ab|N\).%
\end{proof}
As a consequence, the product of two coprime integers divides any common multiple of those integers. This theorem is particularly useful in the context of RSA, where we often deal with products of two coprime integers.%
\begin{theorem}{Theorem}{}{}{co-prime-of-divisors}%
Given any two coprime integers \(P\) and \(Q\) and any two integers \(p\) and \(q\), if \(p|P\) and \(q|Q\), then \(p\) and \(q\) are coprime, i.e., \(\gcd(p, q) = 1\).%
\end{theorem}
\begin{proof}{Proof}{}{co-prime-of-divisors-2}
This theorem can be proved using similar reasoning as in \hyperref[product-of-divisors]{Theorem~{\xreffont\ref{product-of-divisors}}, p.\,\pageref{product-of-divisors}}.%
\end{proof}
\begin{theorem}{Theorem}{}{}{product_of_gcd}%
For any three positive integers \(m\), \(P\), and \(Q\), we have:%
\begin{align*}
(1) \quad & \gcd(m, PQ) \leq \gcd(m, P) \times \gcd(m, Q)\\
(2) \quad & \text{ if } \gcd(P, Q) = 1, \text{ then } \gcd(m, PQ) = \gcd(m, P) \times \gcd(m, Q)
\end{align*}
\begin{proof}{Proof}{}{product_of_gcd-1-3}
(1)%
\par
By the definition of gcd, we have:%
\begin{equation*}
g = \gcd(m, PQ) = \max\{d \mid d \text{ divides } m \text{ and } d \text{ divides } PQ\}
\end{equation*}
Since g divides \(PQ\) and \(m\), g can be expressed as \(p * q\), where \(p\) divides \(P\) and \(m\), and \(q\) divides \(Q\) and \(m\). Therefore,%
\begin{equation*}
p \leq \gcd(m, P) \quad \text{and} \quad q \leq \gcd(m, Q)
\end{equation*}
which implies that%
\begin{equation*}
g = p * q \leq \gcd(m, P) * \gcd(m, Q)
\end{equation*}
\par
This proves part (1).%
\par
(2)%
\par
Let%
\begin{equation*}
p = \gcd(m, P) \quad \text{and} \quad q = \gcd(m, Q)
\end{equation*}
Since \(\gcd(P, Q) = 1\), by \hyperref[co-prime-of-divisors]{Theorem~{\xreffont\ref{co-prime-of-divisors}}, p.\,\pageref{co-prime-of-divisors}}, we have \(\gcd(p, q) = 1\).%
\par
Consequently, by \hyperref[product-of-divisors]{Theorem~{\xreffont\ref{product-of-divisors}}, p.\,\pageref{product-of-divisors}}, we have \(p*q\) divides both \(m\) and \(PQ\). Therefore,%
\begin{equation*}
\gcd(m,P) \times \gcd(m,Q) = p \times q \leq \gcd(m, PQ) 
\end{equation*}
Together with (1), we prove part (2).%
\end{proof}
\end{theorem}
\end{subsectionptx}
\typeout{************************************************}
\typeout{Subsection 2.3 Modular Arithmetc and Congruences}
\typeout{************************************************}
\begin{subsectionptx}{Subsection}{Modular Arithmetc and Congruences}{}{Modular Arithmetc and Congruences}{}{}{subsection-2-3}
Modular arithmetic is a system of arithmetic for integers, where numbers "wrap around" after reaching a certain value, known as the modulus. More specifically, given a positive integer \(N\), we can partition integers into \(N\) equivalence classes denoted by \({0},{1}, \cdots, {N-1}\)  so that each \({i}\) represents the set of integers that have the same remainder \(i\) when they are divided by \(N\).%
\begin{definition}{Definition}{Congruence Relation.}{def-congruence}%
We say that two integers \(a\) and \(b\) are congruent modulo \(N\), denoted as \(a \equiv b \pmod{N}\), if they have the same remainder when divided by \(N\). Equivalently, \(a \equiv b \pmod{N}\) if and only if \(N | (a - b)\).%
\end{definition}
\begin{theorem}{Theorem}{}{}{congruence-theorem}%
We are given a positive integer \(N\) and%
\begin{equation*}
N = P \times Q
\end{equation*}
where \(P\) and \(Q\) are positive integers and \(\gcd(P, Q) = 1\). The following property holds for any two integers \(a\) and \(b\):%
\begin{equation*}
a \equiv b \pmod{N}
\end{equation*}
if and only if%
\begin{equation*}
a \equiv b \pmod{P} \quad \text{and} \quad a \equiv b \pmod{Q}
\end{equation*}
\end{theorem}
\begin{proof}{Proof}{}{subsection-2-3-5}
Based on the fact that%
\begin{equation*}
a \equiv b \pmod{N} \text{ if and only if } N|(a-b)
\end{equation*}
and \hyperref[product-of-divisors]{Theorem~{\xreffont\ref{product-of-divisors}}, p.\,\pageref{product-of-divisors}}, we can easily prove the theorem. \(\square\)%
\end{proof}
This theorem is particularly useful in RSA, where we often work with products of primes and their multiples. It allows us to break down congruences modulo a composite number into simpler congruences modulo its prime factors. This fact is essntial for the proof of correctness of RSA, as it enables us to analyze the encryption and decryption processes separately for each prime factor of the modulus \(N\).%
\end{subsectionptx}
\end{sectionptx}
\typeout{************************************************}
\typeout{Section 3 \(\phi\text{-set}\) and \(\phi\) (Euler's Totient) Function}
\typeout{************************************************}
\begin{sectionptx}{Section}{\(\phi\text{-set}\) and \(\phi\) (Euler's Totient) Function}{}{\(\phi\text{-set}\) and \(\phi\) (Euler's Totient) Function}{}{}{section-3}
\begin{introduction}{}%
In this section, we introduce Euler's Totient Function, which is an important concept in number theory and plays a significant role in the RSA encryption scheme.%
\end{introduction}%
\begin{definition}{Definition}{\(\phi\text{-set}\).}{def-phi-set}%
The \(\phi\text{-set}\) of a positive integer \(N\), denoted as \(\phi\text{-set}(N)\), is the set of all integers from 1 to \(N\) that are coprime to \(N\).%
\end{definition}
For example, if \(N = 10\), then \(\phi\text{-set}(10)\) = \textbraceleft{}1, 3, 7, 9\textbraceright{}, since these are exactly the integers between 1 and 10 that are coprime to 10. In the literature, the \(\phi\text{-set}(N)\) is commonly called the reduced residue system modulo \(N\). It can be verified \(\phi\text{-set}(N)\) forms a group under multiplication modulo \(N\). We use the notation \(\phi\text{-set}(N)\) because we will later extend this set to a bigger set called \(\Phi\text{-set}(N)\), which plays an essential role in the correctness of the RSA encryption scheme.%
\begin{definition}{Definition}{Euler's Totient Function.}{def-euler-totient}%
Euler's Totient Function of a positive integer \(N\), denoted as \(\phi(N)\), is defined as the number of integers from 1 to \(N\) that are coprime to \(N\). As a result \(\phi(N)\) is the number of elements in the \(\phi\text{-set}(N)\).%
\end{definition}
\begin{theorem}{Theorem}{Computing Euler's Totient Function.}{}{Euler-Totient-Theorem}%
\begin{enumerate}
\item{}If \(p\) is a prime number, then \(\phi(p) = p - 1\).%
\item{}If \(p_1, p_2, \ldots, p_k\) are distinct prime numbers, then%
\begin{equation*}
\phi(p_1 p_2 \cdots p_k) = (p_1 - 1)(p_2 - 1) \cdots (p_k - 1)\text{.}
\end{equation*}
\item{}If \(n\) is a positive integer with the prime factorization \(n = p_1^{k_1} p_2^{k_2} \cdots p_m^{k_m}\), then%
\begin{equation*}
\phi(n) = n \prod_{i=1}^{m} \left( 1 - \frac{1}{p_i} \right)\text{.}
\end{equation*}
\end{enumerate}
\end{theorem}
\begin{corollary}{Corollary}{}{}{corollary-euler-totient}%
If \(N = P \times Q\), where gcd(P, Q) = 1, then \(\phi(N) = \phi(P) \cdot \phi(Q)\).%
\end{corollary}
According to the above theorem, we can design an algorithm to compute the \(\phi(N)\) function by using the prime factorization of \(N\). Consequently, the problem of computing \(\phi(N)\) is reduced to the challenge of finding the prime factorization of \(N\). The security of the RSA encryption scheme relies on the computational difficulty of this factorization problem..%
\end{sectionptx}
\typeout{************************************************}
\typeout{Section 4 \(\Phi\text{-set}\) and \(\Phi\) Function}
\typeout{************************************************}
\begin{sectionptx}{Section}{\(\Phi\text{-set}\) and \(\Phi\) Function}{}{\(\Phi\text{-set}\) and \(\Phi\) Function}{}{}{section-4}
\begin{introduction}{}%
In this section, we define the \(\Phi\text{-set}\) of an positive integer \(N\), which is an extension or a super set of \(\phi\text{-set}\) of \(N\) and plays an essential role in the correctness of the RSA encryption scheme.%
\end{introduction}%
\begin{definition}{Definition}{\(\Phi\text{-set}\).}{def-big-phi-set}%
The \(\Phi\text{-set}\) of a positive integer \(N\), denoted as \(\Phi\text{-set}(N)\), is the set of all integers \(m, 1 \leq m \leq N\) that satisfy the following condition:%
\par
\begin{equation*}
\text{if } P = \gcd(m, N) \text{ and } Q = N \div P \text{ then } \gcd(P, Q) = 1\text{.}
\end{equation*}
Or equivalently,%
\begin{equation*}
\Phi\text{-set}(N) = \{m, 1 \leq m \leq N \mid \gcd(P,Q) = 1, \text{where } P = \gcd(m, N) \text{ and } Q = N \div P\}
\end{equation*}
\end{definition}
For example, if \(N = 10\), then \(\Phi\text{-set}(10)\) = \textbraceleft{}1,2,3,4,5,6,7,8,9,10\textbraceright{}. We can easily verify the result manually. Or we can use the theoream \hyperref[Phi-Set-Theorem]{Theorem~{\xreffont\ref{Phi-Set-Theorem}}, p.\,\pageref{Phi-Set-Theorem}} for \(10 = 2 \ldots 5\) to prove the result.%
\begin{remark}{Remark}{}{section-4-5}%
For \(m \in \Phi\text{-set}(N)\), from the definition and \hyperref[product_of_gcd]{Theorem~{\xreffont\ref{product_of_gcd}}, p.\,\pageref{product_of_gcd}}, it is clear that \(gcd(m, Q) = 1\).%
\end{remark}
The \(\Phi\) function of a positive integer \(N\), denoted as \(\Phi(N)\), is defined as the number of integers in \(\Phi\text{-set}(N)\).%
\par
It is clear that \(\Phi\text{-set}(N)\) is a super set of \(\phi\text{-set}(N)\), i.e., \(\Phi\text{-set}(N) \supseteq \phi\text{-set}(N)\). As a result, we have \(\Phi(N) \geq \phi(N)\).%
\begin{theorem}{Theorem}{Computing \(\Phi\) Function.}{}{Phi-Set-Theorem}%
For an positive integer \(N\) if :%
\begin{equation*}
N = p_1 p_2 \cdots p_k
\end{equation*}
where \(p_1,p_2,\ldots,p_k\) are distinct prime numbers, then \(\Phi(N) = N\), i.e. the \(\Phi\text{-set}(N)\) contains all integers from 1 to \(N\).%
\end{theorem}
\begin{proof}{Proof}{}{Phi-Set-Theorem-3}
Let \(m\) be any integer such that \(1 \leq m \leq N\). We need to show that \(m \in \Phi\text{-set}(N)\). Let \(P = \gcd(m, N)\) and \(Q = N \div P\). Since \(N\) is a product of distinct primes and \(P\) is a divisor of \(N\), \(P\) must be a product of a subset of the prime factors of \(N\). Therefore, \(Q = N \div P\) is the product of the remaining prime factors of \(N\). Cosequently, \(P\) and \(Q\) do not share any common prime factors, which implies that \(\gcd(P, Q) = 1\) and as a result \(m \in \Phi\text{-set}(N)\). This concludes the proof.%
\end{proof}
\end{sectionptx}
\typeout{************************************************}
\typeout{Section 5 Euler Theorem}
\typeout{************************************************}
\begin{sectionptx}{Section}{Euler Theorem}{}{Euler Theorem}{}{}{section-5}
\begin{introduction}{}%
In this section, we state the well-known Euler’s Theorem, a fundamental result in number theory related to modular arithmetic. This theorem is essential for understanding the RSA encryption scheme and its generalizations.%
\end{introduction}%
\begin{theorem}{Theorem}{Euler's Theorem.}{}{Euler-Theorem}%
For any two positive integers \(m \text{ and } N\), if \(\gcd(m,N) = 1\), then \(m^{\phi(n)} \equiv 1 \pmod N\).%
\end{theorem}
\begin{proof}{Proof}{}{Euler-Theorem-3}
The proof of Euler's Theorem is equvalant to prove for any \(m \in \phi\text{-set}(N)\),%
\begin{equation*}
m^{\phi(n)} \equiv 1 \pmod N\text{.}
\end{equation*}
\par
Since \(\phi\text{-set}(N)\) forms a group under multiplication modulo \(N\), for any \(m \in \phi\text{-set}(N)\), the cyclic group generated by \(m\) is a subgroup of \(\phi\text{-set}(N)\). Therefore, the order of \(m\), denoted by \(order(m)\), is a divisor of \(\phi(N)\), which is the order of the group \(\phi\text{-set}(N)\).%
\par
Therefore, we have%
\par
\begin{equation*}
m^{order(m)} \equiv 1 \pmod N
\end{equation*}
and%
\begin{equation*}
\phi(N) = k \times order(m), \text{ where } k \text{ is a positive integer.}
\end{equation*}
Combining the above results, we get%
\begin{equation*}
m^{\phi(N)} = m^{k \times order(m)} = (m^{order(m)})^k \equiv 1^k \equiv 1 \pmod N\text{.}
\end{equation*}
\end{proof}
\end{sectionptx}
\typeout{************************************************}
\typeout{Section 6 Extended RSA Encryption Scheme}
\typeout{************************************************}
\begin{sectionptx}{Section}{Extended RSA Encryption Scheme}{}{Extended RSA Encryption Scheme}{}{}{section-6}
\begin{introduction}{}%
In this section, we define the extended RSA encryption scheme and prove its correctness by demonstrating that the decryption process correctly recovers the original plaintext message.%
\end{introduction}%
\begin{definition}{Definition}{Extended RSA Encryption Scheme.}{def-extended-rsa}%
For any positive integer \(N\), the extended RSA encryption scheme is defined as follows: choose \(e\) and \(d\) such that \(1 \leq e, d \leq \phi(N)\) and \(e \times d \equiv 1 \pmod{\phi(N)}\), then the encryption function is defined as:%
\par
For any \(m\), \(1 \leq m \leq N\), \(\text{Enc}(m) = m^e \pmod{N}\)%
\par
and the decryption function is defined as:%
\par
For any \(c\), \(1 \leq c \leq N\), \(\text{Dec}(c) = c^d \pmod{N}\)%
\end{definition}
Like the classical RSA encryption scheme, the correctness condition for the extended RSA encryption scheme is formulated as:%
\par
\begin{equation*}
(m^e)^d \equiv m \pmod{N} \text{ for any integer } m, 1 \leq m \leq N\text{.}
\end{equation*}
\par
The main difference between the extended RSA encryption scheme and the classical RSA encryption scheme is that the extended RSA scheme permits a wider choices of \(N\). The main result of this paper shows the set of messages satisfying the correctness condition of the extended RSA encryption scheme is exactly the \(\Phi\text{-set}(N)\). This main result is proved in the next Theorem.%
\begin{theorem}{Theorem}{Correctness of Extended RSA Encryption Scheme.}{}{extended-rsa-correctness}%
For any positive integer \(N\), the correctness condition of the extended RSA encryption scheme holds for all integers in \(\Phi\text{-set}(N)\).%
\par
In other words, for any integer \(m\) in \(\Phi\text{-set}(N)\), we have:%
\begin{equation*}
(m^e)^d \equiv m \pmod{N}\text{.}
\end{equation*}
\end{theorem}
\begin{proof}{Proof}{}{extended-rsa-correctness-3}
By definition, for any integer \(m\) in \(\Phi\text{-set}(N)\), we have:%
\begin{equation*}
P = gcd(m,N), Q = N \div P, \text{ and }, gcd(P,Q) = 1.
\end{equation*}
We need to prove for such \(m\) that:%
\begin{equation*}
(m^e)^d \equiv m \pmod{N}
\end{equation*}
\par
Since \(e \times d \equiv 1 \pmod{\phi(N)}\), we can express \(e \times d\) as:%
\begin{equation*}
e \times d = k \times \phi(N) + 1
\end{equation*}
for some integer \(k\).%
\par
Now, we can rewrite \((m^e)^d\) as:%
\begin{equation*}
(m^e)^d = m^{k \times \phi(N) + 1}
\end{equation*}
and we can easily verify that:%
\begin{equation*}
m^{k \times \phi(N) + 1} \equiv m \pmod{P}
\end{equation*}
(both sides are divisible by \(P\) due to \(P = gcd(m,N)\)).%
\par
Next, we want to show that \(m^{k \times \phi(N) + 1} \equiv m \pmod{Q}\). Since \(gcd(P,Q) = 1\), \(\phi(N) = \phi(P) \cdot \phi(Q)\) holds and we can rewrite \(m^{k \times \phi(N) + 1}\) as:%
\begin{equation*}
m^{k \times \phi(P) \cdot \phi(Q) + 1} = m^{k \times \phi(P) \cdot \phi(Q)} \cdot m\text{.}
\end{equation*}
\par
Since \(gcd(P,Q) = 1\), which implies \(gcd(m,Q) = 1\) (because \(m \in \Phi\text{-set}(N)\)), we can apply Euler's theorem, which states that if \(gcd(m,Q) = 1\), then:%
\begin{equation*}
m^{\phi(Q)} \equiv 1 \pmod{Q}\text{.}
\end{equation*}
From this we can conclude that:%
\begin{equation*}
m^{k \times \phi(P) \cdot \phi(Q) + 1} \equiv m \pmod{Q}\text{.}
\end{equation*}
\par
Summarize what we have proved:%
\begin{equation*}
(m^e)^d \equiv m \pmod{P}
\end{equation*}
and%
\begin{equation*}
(m^e)^d \equiv m \pmod{Q}\text{.}
\end{equation*}
Since \(gcd(P,Q) = 1\) and \(N = P \times Q\), by \hyperref[congruence-theorem]{Theorem~{\xreffont\ref{congruence-theorem}}, p.\,\pageref{congruence-theorem}}, we have:%
\begin{equation*}
(m^e)^d \equiv m \pmod{N}\text{.}
\end{equation*}
This completes the proof of the theorem.%
\end{proof}
\begin{corollary}{Corollary}{.}{}{corollary-extended-rsa}%
If \(N\) is a positive integer and%
\begin{equation*}
N = p_1 p_2 \cdots p_k
\end{equation*}
where%
\begin{equation*}
p_1, p_2, \ldots, p_k
\end{equation*}
are distinct prime numbers, then the correctness condition of the RSA encryption scheme holds for all \(m\), \(1 \leq m \leq N\).%
\end{corollary}
This corollary follows directly from \hyperref[extended-rsa-correctness]{Theorem~{\xreffont\ref{extended-rsa-correctness}}, p.\,\pageref{extended-rsa-correctness}} and \hyperref[Phi-Set-Theorem]{Theorem~{\xreffont\ref{Phi-Set-Theorem}}, p.\,\pageref{Phi-Set-Theorem}}.%
\par
In the classical RSA algorithm, where \(N = p \times q\) with \(p\) and \(q\) distinct primes, this corollary establishes the correctness of the classical RSA encryption scheme.%
\begin{example}{Example}{Example of classical RSA Encryption Scheme.}{section-6-11}%
Consider \(N = 10 = 2 \times 5\). We can easily compute:%
\begin{equation*}
\phi\text{-set}(N) = \{1,3,7,9\}
\end{equation*}
and%
\begin{equation*}
\Phi\text{-set}(N) = \{1,2,3,4,5,6,7,8,9,10\}\text{.}
\end{equation*}
If we choose \(e = 3\) and \(d = 7\) (since \(3 \times 7 \equiv 1 \pmod{4}\) and \(\phi(10) = 4\)) as the keys for the RSA encryption scheme, then from the corollary, it follows the correctness condition holds for all \(m\), \(1 \leq m \leq 10\).%
\end{example}
\begin{example}{Example}{Example of extended RSA Encryption Scheme.}{section-6-12}%
Consider \(N = 20 = 2^2 \times 5 \). We can easily compute:%
\begin{equation*}
\phi\text{-set}(N) = \{1,3,7,9,11,13,17,19\}
\end{equation*}
and%
\begin{equation*}
\Phi\text{-set}(N) = \{1,3,4,5,7,8,9,11,12,13,15,16,17,19,20\}\text{.}
\end{equation*}
If we choose \(e = 3\) and \(d = 3\) (since \(3 \times 3 \equiv 1 \pmod{8}\) and \(\phi(20) = 8\)), as the keys for the RSA encryption scheme, then we can verify for any \(m\) in \(\Phi\text{-set}(20)\) that the correctness condition holds. However, \textbraceleft{}2,6,10,14,18\textbraceright{} are not in \(\Phi\text{-set}(20)\) and we can verify that all of those integers do not satisfy the correctness condition. In fact,%
\begin{equation*}
(2^3)^3 \equiv 12 \pmod{20}\text{,}
\end{equation*}
\begin{equation*}
(6^3)^3 \equiv 16 \pmod{20}\text{,}
\end{equation*}
\begin{equation*}
(10^3)^3 \equiv 0 \pmod{20}\text{,}
\end{equation*}
\begin{equation*}
(14^3)^3 \equiv 4 \pmod{20}\text{,}
\end{equation*}
\begin{equation*}
(18^3)^3 \equiv 8 \pmod{20}\text{.}
\end{equation*}
None of these results equals the original plaintext message. This example demonstrates that, for the given choices of \(N, e, \text{ and } d\), the messages for which the extended RSA encryption scheme is correct are exactly those in \(\Phi\text{-set}(N)\).%
\end{example}
\end{sectionptx}
\typeout{************************************************}
\typeout{Section 7 Future Work}
\typeout{************************************************}
\begin{sectionptx}{Section}{Future Work}{}{Future Work}{}{}{section-7}
\begin{introduction}{}%
In this section, we will discuss future work related to the generalized RSA encryption scheme.%
\end{introduction}%
\typeout{************************************************}
\typeout{Subsection 7.1 Necessary Condition of the Correctness \(\Phi\text{-set}(N)\)}
\typeout{************************************************}
\begin{subsectionptx}{Subsection}{Necessary Condition of the Correctness \(\Phi\text{-set}(N)\)}{}{Necessary Condition of the Correctness \(\Phi\text{-set}(N)\)}{}{}{subsection-7-1}
We have proved that all messages in \(\Phi\text{-set}(N)\) satisfy the correctness condition of the generalized RSA encryption scheme. However, we did not prove that it is necessary that messages must be in \(\Phi\text{-set}(N)\) to satisfy the correctness condition. Example 6.5 in section 6 demonstrates that the necessary condition for the correctness of the generalized RSA encryption scheme is true. Moreover, we have performed many computational experiments for various values of \(N\) and \(e\) (except \(e = 1\)) to verify that the necessary condition is true for all tested cases. Therefore, it is reasonable to conjecture that the necessary condition for the correctness of the generalized RSA encryption scheme is true. However, a formal proof of this conjecture is still missing and needs to be done in future work.%
\end{subsectionptx}
\typeout{************************************************}
\typeout{Subsection 7.2 Efficient Computing of \(\Phi(N)\)}
\typeout{************************************************}
\begin{subsectionptx}{Subsection}{Efficient Computing of \(\Phi(N)\)}{}{Efficient Computing of \(\Phi(N)\)}{}{}{subsection-7-2}
The computation of \(\Phi(N)\) is another interesting problem. However, the brute-force method for computing \(\Phi(N)\) is not efficient enough for large values of \(N\). In future work, we will explore more efficient algorithms for computing \(\Phi(N)\).%
\end{subsectionptx}
\end{sectionptx}
\typeout{************************************************}
\typeout{References 8 References}
\typeout{************************************************}
\begin{references-section}{References}{References}{}{References}{}{}{references}
\begin{referencelist}
\bibitem[1]{RSA78}\label{RSA78}{}\hypertarget{RSA78}{}Ronald L Rivest, Adi Shamir, and Leonard Adleman, ``A method for obtaining digital signatures and public-key cryptosystems'', in \pubtitle{Communications of the ACM,}, 21(2):120–126, 1978.
\end{referencelist}
\end{references-section}
\end{document}